\documentclass[12pt]{amsart}

\theoremstyle{plain}

\theoremstyle{theorem}
\newtheorem{thm}{Theorem}[section]

\newtheorem{lemma}[thm]{Lemma}
\newtheorem{corollary}[thm]{Corollary}

\theoremstyle{definition}
\newtheorem{definition}[thm]{Definition}
\newtheorem{remark}[thm]{Remark}

\numberwithin{equation}{section}

\title{On a Symmetrization of Diffusion Processes }
\author{Jir\^o Akahori and Yuri Imamura}
\address{Department of Mathematical SciencesCRitsumeikan University \\
1-1-1 Nojihigashi, Kusatsu, Shiga, Japan}
\email{yuri.imamura@gmail.com}

\begin{document}
\maketitle

\begin{abstract}
The latter author, together with collaborators, 
proposed a numerical scheme to calculate
the price of barrier options in \cite{iiko, iio1, iio2}. 
The scheme is based on a symmetrization of diffusion process. 
The present paper aims to give a mathematical credit to the use of 
the numerical scheme for Heston or SABR type 
stochastic volatility models. This will be done by 
showing a fairly general result on the symmetrization 
(in multi-dimension/multi-reflections). 
Further applications (to time-inhomogeneous diffusions/
to time dependent boundaries/to curved boundaries) 
are also discussed. 
\end{abstract}

\section{Introduction}
Recently, the latter author, together with Yuta Ishigaki, 
Takuya Kawagoe and Toshiki Okumura,
introduced a new numerical scheme for calculating 
the price of barrier options in a series of papers \cite{iiko, iio1, iio2}. 
The scheme is based on what they call ``put-call symmetry", 
a notion introduced by Peter Carr and Roger Lee \cite{peter}
in relation to (a generalization of) static hedging of barrier options. 

The put-call symmetry, PCS for short, 
of a diffusion process $ X $ 
at a point $ K \in \mathbf{R} $,  
means %, roughly speaking, 
the equivalence in law 
between $ K - X_t $(put) and $ X_t -K $ (call) 
for arbitrary $ t \geq \tau_K $, where 
$ \tau_K $ is the first hitting time at $ K $. 
This is much weaker than the {\em reflection principle} 
which has been widely-recognized as a fundamental requirement
for the static hedging of barrier options. 
While the put-call symmetry is still something one cannot expect 
for a given diffusion without luck, 
the latter author and her collaborators have noticed that, 
for any given diffusion 
%that is a solution 
%to stochastic differential equation (SDE), 
\begin{itemize}
\item one can construct (easily) another diffusion 
%(as a solution to another SDE) 
that is identical with 
the diffusion up to the first hitting time 
and satisfies the put-call symmetry,  
\item by which the price of barrier-type options 
written on the original diffusion is expressed by 
a combination of the prices of plain options 
written on the constructed one.
\item The fact in turn implies that the ``symmetrization"
offers a new numerical scheme for calculating 
the price of barrier options; it transforms path-dependent 
expectations to path-independent ones. 
\end{itemize}
The numerical experiments they performed show that 
the scheme is quite plausible. They also claim in \cite{iio1,iio2} that
it can be applicable to stochastic volatility models where 
the stock and its volatility are described by a two-dimensional diffusion. 

In this paper, we give a mathematical backgrounds 
for the scheme by establishing some symmetry results 
in a more general setting under the action of 
a reflection group. This in turn leads to further possible applications 
of the scheme. We will show that it can be used for the cases with 
time-inhomogeneous diffusions, 
time-dependent boundaries, as well as curved boundaries. 

This paper is organized as follows. We start with 
a detailed discussion of the put-call symmetry (section \ref{PCS}). 
After recalling the one-dimensional cases (subsection \ref{1dim}), 
we first introduce a multi-dimensional generalization (subsection \ref{MDG})
and then extends it to multi-reflections (subsection \ref{PCSwRG}).   
The key assumptions are found in the statement of Lemma \ref{KEY1}. 
Section \ref{SYMRZ} is devoted to discussions on the `symmetrization". 
Again staring from a review of one-dimensional cases (subsection \ref{1dim2}), 
we give the main result (Theorem \ref{main}) in full generality 
in subsection \ref{FGT}. Applications of the main theorem are presented in 
section \ref{APP}. We first give a credit to the use 
of the scheme for stochastic volatility models 
including Heston's and SABR type (subsection \ref{SVMS}).
The trick for the generic stochastic volatility models is 
extended to get a generalization of Theorem \ref{main} (Theorem \ref{exmain}), 
as a corollary to which we show that the scheme is applicable even to 
time-inhomogeneous diffusions (Corollary \ref{TIH}).  
This observation further enables the applications to 
curved boundary cases (subsection \ref{CVB}) and time-dependent boundary cases 
(subsection \ref{TDB}).

\section{The put-call symmetry}\label{PCS}

\subsection{One dimensional case revisited}\label{1dim}
Let $ X $ be a one dimensional diffusion satisfying 
the put-call symmetry at $ K \in \mathbf{R} $; 
\begin{equation}\label{PCS1}
X_t 1_{ \{ \tau_K \leq t \} }
\overset{\mathop{d}}{=} (2K -X_t) 1_{ \{\tau_K \leq t \} } 
\end{equation}
for any $ t > 0 $, where
$ \tau_K = \inf \{ t > 0 : X_t = K \} $. 
The PCS (\ref{PCS1}) alone, without the reflection principle, 
suffices to have a static hedging formula for barrier options.
For completeness, we give a brief proof on it. Put 
\begin{equation*}\label{SMRZ1}
\tilde{X}_t = X_t 1_{ \{ \tau_K > t \} }
+ (2K-X_t) 1_{ \{ \tau_K \leq t \} }. 
\end{equation*}
Then by the PCS (\ref{PCS1}), we have
\begin{equation}\label{EQL1}
X_t \overset{\mathop{d}}{=} \tilde{X}_t, 
\end{equation}
for arbitrary $ t > 0 $.
The equivalence in law (\ref{EQL1}) in turn implies  
\begin{equation}\label{EQL2}
P ( X_t \in A : \tau_K \leq t) 
= P ( \tilde{X}_t \in A: \tau_K \leq t ) 
\end{equation}
for any Borel set $ A $. 
Suppose that  $ X_0 > K $ and $ A \subset \{x>K\} $.
Then the right-hand-side is equal to 
$ P ( 2K-X_t \in A ) $ since $ \{ \tau_K \leq t \} $ 
is included in $ \{ 2K-X_t \in A\} $. 
Thus it holds that 
\begin{equation}\label{STH}
P ( X_t \in A : \tau_K > t) = 
P (X_t \in A ) - P (X_t \in 2K -A ) 
\end{equation}
for any Borel set $ A $ with $ A \subset \{x>K\} $
(see \cite{peter}, \cite{iiko}).
In other words,  
\begin{equation}\label{STH2}
\begin{split}
& E [ f( X_t-K )1_{\{X_t>K\}} 1_{\{ \tau_K > t\}}]  \\
& \qquad = 
E[f(X_t-K)1_{\{X_t>K\}} ] - E [f (K-X_t)1_{\{X_t<K\}}]
\end{split}
\end{equation}
for any bounded Borel function $ f $ and $ t>0 $, which can be understood 
as a static hedging formula. 

\subsection{A multi-dimensional generalization of PCS}\label{MDG}
%with respect to a finite reflection group}
To generalize the argument in the previous subsection, 
we understand $ x \mapsto 2k- x $ as a reflection. 
In $ \mathbf{R}^d $, a reflection is associated 
with a hyperplane.
A hyperplane is 
given by 
$ H_{\alpha,k} := \{ x \in \mathbf{R}^d : \langle \alpha, x \rangle = k  \} $,
where $ k \in \mathbf{R} $ and $ \alpha (\ne 0) \in \mathbf{R}^d $. 
The reflection $ s_{\alpha, k} : \mathbf{R}^d \to \mathbf{R}^d $
with respect to $ H_{\alpha,k} $ is given by
\begin{equation*}\label{RFL1}
s_{\alpha, k} (x) 
=  x - ( \langle x,\alpha \rangle -k ) \frac{2 \alpha}{
|\alpha|^2}.
\end{equation*}
Notice that $ s_{\alpha, k}^2 = \mathrm{Id}_{\mathbf{R}^d} $. 
A natural extension of the previous PCS could be as follows. 
\begin{definition}
Let $ X $ be a diffusion process in $ \mathbf{R}^d $. 
We say $ X $ has the put-call symmetry with respect to the 
hyperplane $ H_{\alpha, k} $ if 
\begin{equation}\label{PCSd}
X_t 1_{ \{ \tau_{\alpha,k} \leq t \} }
\overset{\mathop{d}}{=} s_{\alpha, k}(X_t)
1_{ \{\tau_{\alpha,k}  \leq t \} } 
\end{equation}
for any $ t > 0 $, where
$ \tau_{\alpha,k}  = \inf \{ t > 0 : X_t = H_{\alpha,k} \} $. 
\end{definition}
In a totally similar way as above, we can obtain a
static hedging formula corresponding to (\ref{STH2}). 
In fact, we have, for $ t > 0 $, 
\begin{equation}\label{STH3}
E [ f( X_t )1_{\{ \tau_{\alpha, k} > t\}}]  = 
E[f(X_t)] - E [f (s_{\alpha, k}(X_t))]. 
\end{equation}
for any bounded measurable $ f $ whose support is included in 
a half space
$ \{ x \in \mathbf{R}^d | \pm (\langle \alpha, x \rangle - k) > 0 \} $,
provided that $ X $ has the PCS
with respect to a hyperplane $ H_{\alpha, k} $. 

\subsection{The PCS with respect to a reflection group}\label{PCSwRG}
As the reflection principle is generalized to multiple reflections 
which form a group (see \cite{IT} and references therein), 
so is the put-call symmetry. In this section, 
we discuss the generalization in detail.

It may be natural that we consider the exit time out of an 
intersection of hyperplanes. More precisely, denoting
\begin{equation*}
H_{\alpha, k}^+ := \{ x \in \mathbf{R}^d | \langle \alpha, x \rangle - k > 0 \}
\end{equation*}
and 
\begin{equation*}
\Sigma_{\Phi} := \bigcap_{ (\alpha,k) \in \Phi } H_{\alpha, k}^+, 
\end{equation*}
we set, for a given diffusion $ X $, 
\begin{equation*}
\tau_{\Sigma_\Phi} := 
\inf \{ t>0 : X_t \not\in \Sigma_{\Phi} \}
=\min_{(\alpha,k) \in \Phi} \tau_{\alpha,k} ,
\end{equation*}
and consider the problem of representing 
the expectation of $ f(X_t)1_{\{ \tau_{\Sigma_\Phi} >t\} } $
by those of $ f ( g (X_t) ) $, where $ g $ runs through 
a set $ G $, which will turn out to be 
the group generated by the reflections 
$ s_\lambda $, $ \lambda \in \Phi $.

Looking at the discussion in section \ref{1dim},
we notice that the key was the equation (\ref{EQL2}), 
which is not anymore directly applicable to the multi-reflection case.
However, we have the following generalization.
\begin{lemma}\label{KEY1}
Let $ G $ be the group generated by the reflections $ \{s_{\alpha,k}: 
(\alpha,k) \in \Phi \} $ and $ X $ be a diffusion process in $ \mathbf{R}^d $
satisfying PCS with respect to $ H_{\alpha,k} $ for all 
$ (\alpha, k) \in \Phi $. 
Assume that (i) $ g \Sigma \cap g' \Sigma = \emptyset $ 
whenever $ g \ne g' \in G $, and (ii) there is a group homomorphism 
$ \eta : G \to \mathbf{C} $ (character) such that 
$ \eta (s_{\alpha,k}) = -1 $ for each $ (\alpha, k) \in \Phi$. 
Then, for a Borel subset $ A $ of $ \Sigma $, $ x \in \Sigma $,
and $ t > 0 $, we have
\begin{equation}\label{GRP}
\sum_{g \in G} \eta (g) P_x (X_t \in g A : \tau_{\Sigma_\Phi} \leq t )
= 0. 
\end{equation}
\end{lemma}
\begin{proof}
%Set 
%\begin{equation*}
%X^{\alpha, k}_t = X_t 1_{ \{ \tau_{\alpha, k} > t \} }
%+ s_{\alpha, k} (X_t) 1_{ \{ \tau_{\alpha, k} \leq t \} }. 
%\end{equation*}
We first note that the left-hand-side of (\ref{GRP}) is absolutely 
convergent since the sets $ g A, g \in G $ are disjoint. 
Therefore we can change the order as 
\begin{equation}\label{CTO1}
\begin{split}
& \sum_{g \in G} \eta (g) P_x (X_t \in g A : \tau_{\Sigma_\Phi} \leq t ) \\
& \qquad 
= \sum_{(\alpha,k) \in \Phi } 
\sum_{g \in G} \eta (g) P_x (X_t \in g A : \tau_{\Sigma_\Phi}= \tau_{\alpha,k} 
\leq t ). 
\end{split}
\end{equation}
By the assumption on the put-call symmetry, we have
\begin{equation}\label{EQL4}
\begin{split}
P_x ( X_t \in g A : \tau_{\Sigma_\Phi}=\tau_{\alpha, k} \leq t) 
%= P ( X^{\alpha, k}_t \in g A: \tau_{\Sigma_\Phi}=\tau_{\alpha, k} \leq t ) \\
&= P_x ( s_{\alpha, k} (X_t) \in g A: \tau_{\Sigma_\Phi}=\tau_{\alpha, k} \leq t ) \\
&= P_x (X_t \in s_{\alpha,k} g A: \tau_{\Sigma_\Phi}=\tau_{\alpha, k} \leq t ).
\end{split}
\end{equation}
%where $ A $ is a Borel subset of $ \Sigma_\Phi $, $ g \in G $ 
%and $ x \in  \Sigma $. 
On the other hand, by the assumption on $ \eta $ we have
\begin{equation*}
\begin{split}
& \sum_{g \in G} \eta (g) 
P_x (X_t \in s_{\alpha,k} g A: \tau_{\Sigma_\Phi}=\tau_{\alpha, k} \leq t ) \\
&\qquad = - \sum_{g \in G}  \eta (s_{\alpha,k} g) 
P_x (X_t \in s_{\alpha,k} g A: \tau_{\Sigma_\Phi}=\tau_{\alpha, k} \leq t ),
\end{split}
\end{equation*}
which is equal to 
\begin{equation*}
- \sum_{g \in G} \eta (g) 
P_x (X_t \in g A: \tau_{\Sigma_\Phi}=\tau_{\alpha, k} \leq t ),
\end{equation*}
thanks to the group structure. This observation together with 
the equation (\ref{EQL4}) shows that 
\begin{equation*}
\sum_{g \in G} \eta (g) 
P_x (X_t \in g A: \tau_{\Sigma_\Phi}=\tau_{\alpha, k} \leq t ) = 0, 
\end{equation*}
which proves the assertion with (\ref{CTO1}). 
\end{proof}

Thanks to the lemma, we have a generalization of 
the static hedging formula.
\begin{thm}
We keep the notations and the assumptions 
of the Lemma \ref{KEY1}. 
For a bounded Borel function $ f $ with its support in $ \Sigma_\Phi $,
we have
\begin{equation}\label{STH4}
E [ f( X_t )1_{\{ \tau_{\Sigma_\Phi} > t\}}]  = 
\sum_{g \in G} \eta (g) E [f (g^{-1}(X_t))].
\end{equation}
\end{thm}
\begin{proof}
It suffices to show that 
\begin{equation*}
P_x ( X_t \in A : \tau_{\Sigma_\Phi} > t )
= \sum_{g \in G} \eta (g) P_x  (X_t \in g A ),
\end{equation*}
but this is equivalent to (\ref{GRP}) since
\begin{equation*}
P_x  (X_t \in g A : \tau_{\Sigma_\Phi} > t) = 0 
\end{equation*}
unless $ g \ne 1 $, by the assumption (ii) in Lemma \ref{KEY1}. 
\end{proof}

\begin{remark}
The one reflection case in section \ref{MDG} automatically 
satisfies the assumptions in Lemma \ref{KEY1}, and 
apparently (\ref{STH4}) includes (\ref{STH3}). 
\end{remark}

\begin{remark}
The case with $ \Phi = \{ (\alpha, k), (-\alpha, -k-1) \} $
corresponds to the double boundary reflections, where, 
by choosing conventionally $ |\alpha|=1 $
and $ \eta $ to be the (mod $ 2 $)``length" of the reflections (see e.g. \cite{HUMP}), 
\begin{equation*}
P_x ( X_t \in A ; \tau_{\Sigma_\Phi} > t )
= \sum_{k \in \mathbf{Z}} \{ P_x (X_t \in A + 2k \alpha ) - 
P_x (X_t \in -A + (2k+1) \alpha )\}. 
\end{equation*}   
\end{remark}

\begin{remark}
More generally, in the cases of multi-reflections, there do exist 
the situations where the assumptions are fulfilled. 
There are two important classes. One is that of 
finite reflection groups  
and the other, of affine reflection groups. 
The domain $ \Sigma $ is a cone in the former class and 
a simplex in the latter. 
In fact, if $ \{ \alpha \} $ forms a so-called 
{\em simple} system (or fundamental system) of a {\em root} system,
and (i) if $ k $ is fixed, then the group becomes finite, 
and (ii) if additionally $ \alpha $ is taken from a root system 
properly and $ k' $ is taken to be $ -k-1 $, then the group 
becomes (isomorphic) to the semi-direct product of the finite reflection
group and the translation group $ \mathbf{Z} $, which is called an
{\em affine reflection group}. In both cases we can take $ \eta $ 
to be its determinant 
when the finite reflection group is embedded 
into the orthogonal group and therefore into the general linear group. 
For details, see e.g. \cite{HUMP}.
%In the present paper, however, we do not need the detailed study of 
%the structure (nor the classification!) of the groups, but 
%just the existence suffices. We will see this in section \ref{SYMRZ}.
\end{remark}

\section{The symmetrization}\label{SYMRZ}
Let
$ X $ be a solution to
\begin{equation}\label{SDE}
dX_t = \sigma (X_t) \,dW_t + \mu(X_t) \,dt,
\end{equation}
where 
$ \sigma : \mathbf{R}^d \to \mathbf{R}^d \times \mathbf{R}^d $ 
and $ \mu: \mathbf{R}^d \to \mathbf{R}^d $ are 
piecewise continuous functions with at most linear growth, and 
$ W $ denotes a $ d $-dimensional standard Wiener process. 
We further impose some ellipticity condtions on $ \sigma \sigma^* $
to ensure that 
a unique (weak) solution to (\ref{SDE}) exists,
and that Euler-Maruyama approximation of the solution of (\ref{SDE}) works;

\subsection{One dimensional case revisited}\label{1dim2}
In the case of $ d =1$ (with $ \alpha=1 $), 
Carr and Lee \cite{peter} showed that
\begin{equation}\label{C1PCS}
 \sigma (x) = \pm \sigma (2k-x), \quad \mu (x) = - \mu (2k-x)
\end{equation}
is a sufficient condition under which 
the solution $ X $ of (\ref{SDE}) has the PCS. 
As is pointed our in Introduction, even though this is much weaker than 
the reflection principle which can be now rephrased as 
``PCS at any k", yet it is not practical to assume 
a price process to satisfy PCS. 
% of a one dimensional diffusion process which is
%given by a solution to an stochastic differential equation. 

In \cite{iiko}, however, the following observation is made;
\begin{thm}[\cite{iiko}]\label{iiko}
Let $ X $ be a solution to (\ref{SDE}) without (\ref{C1PCS})
and $ \tilde{X} $ be one to
\begin{equation*}
d\tilde{X}_t = \tilde{\sigma} (\tilde{X}_t) \,dW_t 
+ \tilde{\mu}(\tilde{X}_t) \,dt,
\end{equation*}
where 
\begin{equation*}
\tilde{\sigma} (x) = \sigma (x) 1_{ \{x > k\} } 
\pm \sigma (2k-x) 1_{ \{x \leq k\} },
\end{equation*}
and 
\begin{equation*}
\tilde{\mu} (x) = \mu (x) 1_{ \{x > k\} } 
- \mu (2k-x) 1_{ \{x \leq k\} }.
\end{equation*}
We assume $ X_0 = \tilde{X}_0 > K $. Then we have
\begin{equation}\label{STH5}
\begin{split}
& E [ f( X_t-K )1_{\{X_t>K\}} 1_{\{ \tau_K > t\}}]  \\
& \qquad = 
E[f(\tilde{X}_t-K)1_{\{\tilde{X}_t>K\}} ] 
- E [f (K-\tilde{X}_t)1_{\{\tilde{X}_t<K\}}]
\end{split}
\end{equation}
for any bounded Borel function $ f $ and $ t>0 $. 
\end{thm}
They claimed that the formula (\ref{STH5}) gives a new insight to 
financial engineering of barrier options; it says 
that the price of barrier option 
is expressed by those of plain options. Numerical analysis for 
the former is difficult, while the latter is much easier.  
In fact, the numerical experiments they performed (part of which is
appeared in (\cite{iiko})) show that their scheme is quiet effective. 

\begin{remark}
The PCS method can be seen as a diffusion equation counterpart 
(or simply a generalization) of 
so-called {\em method of image charges} in electrostatics. 
The authors thank Prof. Sergey Nadtochiy for 
suggesting to us this ``synchronicity". 
\end{remark}
 
\subsection{Symmetrization with respect to a reflection group}\label{FGT}
They also applied their scheme to stochastic volatility models like 
Heston or SABR type, where above Theorem \ref{iiko} cannot be applied directly
since they are basically two dimensional models. 
With a view to endowing a certificate, we give a general result 
on multi-dimensional multi-reflection cases in this section. 
The certificate will be provided in the next section as a corollary
to the result. 

We start with a lemma. 
\begin{lemma}
Let $ A $ be an affine transformation in $ \mathbf{R}^d $ such that 
$ A x = A_0 x+ a $ for $ x \in \mathbf{R}^d $ with 
$ A_0 \in GL (d,\mathbf{R}) $
and $ a \in \mathbf{R}^d $. Suppose that
\begin{equation}\label{SMTRY}
\sigma (A x ) = A_0 \sigma (x) U_x, \quad \mu (Ax) = A_0 \mu (x)
\end{equation}
for $ x \in \mathbf{R}^d $, with some piecewise continuous 
$ x \mapsto U_x \in O (d) $. 
Then $ A X_t $ starting from $ A x $ 
is identically distributed as $ X_t $ starting from $ x \in \mathbf{R}^d $
as a stochastic process provided that $ Ax = x $,
where $ X $ is the unique weak solution to (\ref{SDE}). 
\end{lemma}

\begin{proof}
Put $ Y = A X $. Then, 
\begin{equation*}
d Y_t = A_0 d X_t = A_0 \sigma (X_t) dW_t + A_0 \mu (X_t) dt,
\end{equation*}
which equals to 
\begin{equation*}
\sigma (A X_t) \, U_x^{-1} d{W}_t + \mu (A X_t) \, dt,
\end{equation*}
by the assumptions (\ref{SMTRY}), where we note that $ U_x^{-1} {W}_t $ 
is another Wiener process. Namely, $ Y $ is a weak solution to
\begin{equation*}
d Y_t = \sigma (Y_t) \, dW_t + \mu (Y_t) \, dt.
\end{equation*}
By the uniqueness of the solution, we have the assertion. 
\end{proof}
As in the one dimensional case, we {\em symmetrize} 
both $ \sigma $ and $ \mu $ to get another diffusion 
with the PCS for which an extended static hedging relation 
still holds. 

\begin{thm}\label{main}
Suppose that we are given a family of hyperplanes indexed by $ \Phi $
which satisfies the assumptions of Lemma \ref{KEY1}.
Put 
\begin{equation}\label{SMTRZ}
\tilde{\sigma} (x) = \sum_{g \in G} T_g \sigma (g^{-1} x ) 
1_{ \{x \in g \Sigma_\Phi\}} U_x , 
\quad
\tilde{\mu} (x) = \sum_{g \in G} T_g \mu (g^{-1} x )
1_{ \{x \in g \Sigma_\Phi\}},
\end{equation}
where $ T_g  $ is an orthogonal matrix corresponding 
to the reflection part of $ g $, and $ x \mapsto U_x \in O (d) $ 
is a piecewise continuous map. 
Let $ \tilde{X} $ be a solution to 
\begin{equation*}
d \tilde{X}_t = \tilde{\sigma} (\tilde{X}_t)\,dW_t 
+ \tilde{\mu} (\tilde{X}_t)\,dt. 
\end{equation*}
Then for a bounded Borel function $ 
f $ with its support in $ \Sigma_\Phi $,
we have
\begin{equation}\label{STH6}
E [ f( X_t )1_{\{ \tau_{\Sigma_\Phi} > t\}}]  = 
\sum_{g \in G} \eta (g) E [f (g^{-1}(\tilde{X}_t))].
\end{equation}
Here we used the notations set in section \ref{PCSwRG}. 
\end{thm}
\begin{proof}
It suffices to show that $ \tilde{\sigma} $ and $ \tilde{\mu} $
satisfy (\ref{SMTRY}). 
For $ h \in G $, we can find $ h_0 \in \mathbf{R}^d $ 
and $ T_h \in O (d) $ such that $ h x = T_h x + h_0 $. 
We then have 
\begin{equation*}
\begin{split}
\tilde{\sigma} (hx) 
&= \sum_{ g \in G } T_h T_h^{-1}T_g  \sigma (g^{-1} h x) 
1_{ \{ x \in h^{-1} g \Sigma_\Phi\}} U_{h x}\\
&= T_h \sum_{ g \in G } T_{h^{-1} g } 
\sigma ( (h^{-1} g)^{-1} x)
1_{ \{ x \in h^{-1} g \Sigma_\Phi\}} U_{h x} \\
&= T_h \tilde{\sigma} (x) U_{h x}.
\end{split}
\end{equation*} 
Here we have used $ T_f T_g = T_{f g} $ for $ f, g \in G $, 
which can be verified by
\begin{equation*}
f g (x) = f (T_g x + g_0) = T_f T_g x + T_f g_0 + f_0,
\end{equation*}
where $ f x = T_f x + f_0 $ and $ g x = T_g x + g_0 $. 

Similarly we have $ \tilde{\mu} (hx) = T_h \mu (x) $. 
\end{proof}

\begin{remark}\label{PRCLR}
We can use the numerical scheme based on the equation (\ref{STH6})  
for general $ \Phi $ since, by a linear transformation, 
a generic $ \Phi $ (and given $ X $) can be 
transformed into the system satisfying the assumptions in Lemma
\ref{KEY1}. 
%, with which the expectation with respect to the original 
%$ X $ is easily obtained by $ F^{-1} $. 
\end{remark}

\section{Applications}\label{APP}

\subsection{Stochastic volatility models}\label{SVMS}
As we have stated in the previous section, we give a certificate
to the use of the PCS symmetrization method to 
such stochastic volatility models as Heston's and SABR type. 

A generic stochastic volatility model is given as follows:
\begin{equation}\label{GSV1}
\begin{split}
dX_t &= \sigma_{11} (X_t,V_t) dW_t  + \mu_1 (X_t,V_t) \,dt  \\
dV_t &= \sigma_{21} (V_t) dW_t + \sigma_{22} (V_t) dB_t + \mu_2 (V_t) \, dt,
\end{split}
\end{equation}
where $ W $ and $ B $ are mutually independent ($1$-dim) Wiener processes, 
$ \mu (x,v) = (\mu_1 (x,v), \mu_2 (v) ) $ 
is a continuous function 
on $ \mathbf{R}^2 $, 
and 
\begin{equation*}
\sigma (x, v) = \begin{pmatrix}
\sigma_{11} (x,v) & 0 \\
\sigma_{21} (v) & \sigma_{22} (v)  \\
\end{pmatrix}
\end{equation*}
is a continuous function on $ \mathbf{R}^2 $.
In most cases, $ \sigma_{11} (x,v)  = x \nu (v) $ for some $ \nu $ 
and $ \mu_1 (x,v) = r x $ from financial reasoning, and 
$ x $- and $ v $- axes are set to be natural boundary. 
In our scheme, however, these features are irrelevant but
that $ V $ do not depend on $ S $ is important. 
%The symmetrization of the above $ (S,V) $ is given by the solution to

\begin{corollary}\label{SVPCS}
Let $ K > 0 $ and put 
\begin{equation*}
\tilde{\sigma}_{11} (x,v) = 
\begin{cases}
\sigma_{11} (x,v) & x \geq K \\
- \sigma_{11} (2K-x,v) & x < K
\end{cases}.
\end{equation*}
\begin{equation*}
\tilde{\mu}_{1} (x,v) = 
\begin{cases}
\mu_{1} (x,v) & x \geq K \\
- \mu_{1} (2K-x,v) & x < K 
\end{cases},
\end{equation*}
and let $ \tilde{X} $ be the unique (weak) solution to 
\begin{equation*}
d\tilde{X}_t = \tilde{\sigma}_{11} (\tilde{X}_t,V_t) dW_t  
+ \tilde{\mu}_1 (\tilde{X}_t,V_t) \,dt,
\end{equation*}
%starting from $ x > k $, 
where $ V $ is the solution to (\ref{GSV1}). 
Then, it holds for any bounded Borel function $ f $ and $ t>0 $ that
\begin{equation*}\label{STH7}
\begin{split}
& E [ f( X_t-K )1_{\{X_t>K\}} 1_{\{ \tau_K > t\}}]  \\
& \qquad = 
E[f(\tilde{X}_t-K)1_{\{\tilde{X}_t > K\}} ] 
- E [f (K-\tilde{X}_t)1_{\{ \tilde{X}_t < K \}}],
\end{split}
\end{equation*}
where $ X $ is the solution to (\ref{GSV1}) with $ X_0 > K $ 
and $ \tau_K $ is the first hitting time of $ X $ to $ K $. 
\end{corollary}

\begin{proof}
Apply Theorem \ref{main} to $ \Phi = \{ ( (1,0), K ) \} $, 
with $ U_x \equiv I $.
\end{proof}

We can also obtain a static hedging formula for a 
double barrier option. Let $ K' $ be a positive constant.
We consider static hedging of an option knocked out 
if $ X $ of (\ref{GSV1}) hit either the boundary $ x=K $ or
$ x = K + K' $.      
\begin{corollary}
Set 
\begin{equation*}
\begin{split}
& \hat{\sigma}_{11} (x,v) \\
& = 
\begin{cases}
\sigma_{11} (x-2nK',v) &  K + 2n K' \leq x < k + 
(2n+1) K',  \\
- \sigma_{11} (2K-x+ 2nK',v) & K + (2n-1) K' \leq x < 
K + 2n K',
\end{cases} \\
& \hspace{8cm} n \in \mathbf{Z},
\end{split}
\end{equation*}
\begin{equation*}
\begin{split}
& \hat{\mu}_{1} (x,v) \\
& = 
\begin{cases}
\mu_{1} (x-2nK',v) &  K + 2n K' \leq x < k + 
(2n+1) K',  \\
- \mu_{1} (2K-x+ 2nK',v) & K + (2n-1) K' \leq x < 
K + 2n K',
\end{cases} \\
& \hspace{8cm} n \in \mathbf{Z},
\end{split}
\end{equation*}
and let $ \hat{X} $ be the unique (weak) solution to 
\begin{equation*}
d\hat{X}_t = \hat{\sigma}_{11} (\hat{X}_t,V_t) dW_t  
+ \hat{\mu}_1 (\hat{X}_t,V_t) \,dt,
\end{equation*}
%starting from $ x > k $, 
where $ V $ is the solution to (\ref{GSV1}). 
Then, it holds for any bounded Borel function $ f $ and $ t>0 $ that
\begin{equation*}\label{STH7}
\begin{split}
& E [ f( X_t-K )1_{\{X_t \in (K,K+K') \}} 1_{\{ \tau_{(K,K+K')} > t\}}]  \\
& \quad = 
\sum_{n \in \mathbf{Z} }
E[f(\hat{X}_t-K-2nK')1_{\{\hat{X}_t \in (K,K+K') \}} ] \\
& \qquad - \sum_{n \in \mathbf{Z} }
E [f (K-\hat{X}_t+(2n-1)K')1_{\{ \hat{X}_t \in (K, K+K') \}}],
\end{split}
\end{equation*}
where $ X $ is the solution to (\ref{GSV1}) with $ X_0 \in (K,K+K') $ 
and $ \tau_{(K,K+K')} $ is the first exit time of $ X $ out of $ (K,K+K') $. 
\end{corollary}

\begin{proof}
Apply Theorem \ref{main} to $ \Phi = \{ ( (1,0), K ), ((-1,0),K+K') \} $, 
with $ U_x \equiv I $.
\end{proof}

\subsection{Time-inhomogeneous extension}
The trick of Corollary \ref{SVPCS} further 
leads to a time-inhomogeneous extension 
of Theorem \ref{main}. Even 
more generally, we have the following 
\begin{thm}\label{exmain}
We reset $ X $ and $ V $ in (\ref{GSV1}) as follows:
$ X $ is $ \mathbf{R}^{d_1} $ valued and $ V $ is $ \mathbf{R}^{d_2} $-valued, 
with $ \sigma_{11}: \mathbf{R}^{d_1} \times \mathbf{R}^{d_2} \to 
\mathbf{R}^{d_1} \otimes \mathbf{R}^{d_1} $, 
$ \sigma_{2i} : \mathbf{R}^{d_2} 
\to \mathbf{R}^{d_2} \otimes \mathbf{R}^{d_i} $, $ i=1,2 $, 
$ \mu_1: \mathbf{R}^{d_1} \otimes \mathbf{R}^{d_2} 
\to \mathbf{R}^{d_2} $, and
$ \mu_2: \mathbf{R}^{d_2} \to \mathbf{R}^{d_2} $, all piecewise continuous,
and $ W $, $ d_1 $-dimensional, and $ B $, $ d_2 $-dimensional, 
are mutually independent Wiener processes.  
For the SDE for $ V $, we simply assume that $ \sigma_{2i} $,$ i=1,2 $
and $ \mu_2 $ are given so that there is a unique strong solution. 
Suppose that $ \Phi_0 \subset \mathbf{R}^{d_1} \otimes \mathbf{R} $ 
satisfies the assumptions of Lemma \ref{KEY1}. Let $ G $ be 
the reflection group associated with $ \Phi $. Then we still have 
(\ref{STH6})\footnote{
Precisely speaking, the group $ G $ should be embedded in Affine group 
in $ \mathbf{R}^{d_1 + d_2} $}.
\end{thm}

\begin{proof}
Set $ \Phi = \{ ( (\alpha, 0),k) \in \mathbf{R}^{d_1 + d_2} :
(\alpha, k) \in \Phi_0 \} $, and apply Theorem \ref{main}.  
\end{proof}

\begin{corollary}\label{TIH}
Theorem \ref{main} is generalized to time-inhomogeneous cases:
$ \sigma $ and $ \mu $ can be time-dependent.
\end{corollary}
\begin{proof}
Take $ d_2 = 1 $ and $ \sigma_{21} = \sigma_{22} = 0 $ (and $ \mu_2 = 1 $)
in Theorem \ref{exmain}.   
\end{proof}
\subsection{Curved boundary}\label{CVB}
Let $ D \subset \mathbf{R}^d $ be a domain 
homeomorphic to a $ \Sigma_{\Phi} $
satisfying the assumptions in Lemma \ref{KEY1}. 
The scheme can be applied to the numerical 
calculation of such an expectation as  
\begin{equation*}
E [ f (X_t) 1_{ \{X_t \in D\}} 1_{\{ \tau_D >t\}}],
\end{equation*}
where 
\begin{equation*}
\tau_D := \inf \{ t> 0 : X_t \not\in D \}. 
\end{equation*}
In fact, at least heuristically, 
$ F: D \to \Sigma_\Phi $ being a homeomorphism, 
%and $ \tilde{F} $ is its appropriate extension to $ \mathbf{R}^d $
$ Y_t = {F} (X_t) $ can be realized as a solution to a
time-inhomogeneous SDE in $ \Sigma_{\Phi} $, 
and we can apply Corollary \ref{TIH}. Note that in the scheme 
the information of the outside of $ \Sigma_\Phi $ of 
the coefficients of SDE is totally irrelevant. 
%We do not discuss the topic in a rigorous manner in this paper. 

\subsection{Time-dependent boundary}\label{TDB}
The scheme is also applicable to 
the cases where knock-out boundaries are time-dependent. 

Let $ \alpha_i : [0,\infty) \to \mathbf{R}^d $, $ i=1, \cdots, d $ 
be smooth maps such that the matrix 
$ A (t) = [\alpha_1 (t), \cdots, \alpha_d (t) ] $ 
is invertible for all $ t $ and is at most linear growth in $ t $. 
For each $ t \geq 0 $, put
$ \Phi (t) = \{ (\alpha_i(t), k_i ) \} $\footnote{Without loss of generality, 
we may assume that $ k_i $, $ i=1,2, \cdots, d $ are constants.} and 
let $ \Sigma (t) $ be the intersection of the half spaces 
associated with $ \Phi (t) $. We will briefly discuss how 
it can be transformed to the previous problem, though heuristically. 

Let $ C : [0,\infty) \to \mathbf{R}^d \otimes \mathbf{R}^d $ be 
a solution (unique up to the initial point) to 
the following matrix valued linear differential equation:
\begin{equation}\label{ODE1}
C' (t) = C_t A'(t) \{A (t) \}^{-1}.  
\end{equation}
Then $ C (t) \alpha_i (t) = 
C(0) \alpha_i (0) =: \alpha_i $ for $ i=1,2, \cdots, d $ 
since (\ref{ODE1}) is equivalent to 
\begin{equation*}
(C(t) A(t))' = 0.
\end{equation*}
Moreover, since
\begin{equation*}
\langle \alpha_i (t), x \rangle = 
\langle C (t) \alpha_i (t), C^* (t)^{-1} x \rangle 
= \langle \alpha_i,C^* (t)^{-1} x \rangle, 
\end{equation*}
we see that 
$ x \mapsto C^* (t)^{-1} x $ transforms $ \Sigma (t) $ 
to a time-independent domain $ \Sigma_{\Phi^*} $ with
$ \Phi^* = \{ ( \alpha_i, k_i), i=1,2, \cdots, d \} $. 
Now the problem turns into the pricing of options 
written on $  C^* (t)^{-1} X_t $ knocked out at the boundary of 
$ \Sigma_{\Phi^*} $. 
As we have commented in Remark \ref{PRCLR}, 
we can assume that $ \Phi^* $ satisfies 
the assumptions in Lemma \ref{KEY1}. 
Combining them all, we see that 
Corollary \ref{TIH} is now applicable.

\end{document}